\documentclass[pdflatex,sn-mathphys-num]{sn-jnl}

\usepackage{graphicx}%
\usepackage{multirow}%
\usepackage{amsmath,amssymb,amsfonts}%
\usepackage{amsthm}%
\usepackage{mathrsfs}%
\usepackage[title]{appendix}%
\usepackage{xcolor}%
\usepackage{textcomp}%
\usepackage{manyfoot}%
\usepackage{booktabs}%
\usepackage{algorithm}%
\usepackage{algorithmicx}%
\usepackage{algpseudocode}%
\usepackage{listings}%
\usepackage{todonotes}

\theoremstyle{thmstyleone}%
\newtheorem{theorem}{Theorem}
\newtheorem{proposition}[theorem]{Proposition}%

\theoremstyle{thmstyletwo}%
\newtheorem{example}{Example}%
\newtheorem{corollary}{Corollary}%
\newtheorem{lemma}{Lemma}%

\theoremstyle{thmstylethree}%
\newtheorem{definition}{Definition}%

\begin{document}

\title[Weight Distributions of  Single Parity-Check Product Codes via Character Sums]{Weight Distributions of  Single Parity-Check Product Codes via Character Sums}


\author[1]{\fnm{Makson Miller } \sur{Ribeiro Alves}}\email{m226079@dac.unicamp.br}

\author*[2]{\fnm{Sara} \sur{D. Cardell}}\email{sd.cardell@unesp.br}
\equalcont{These authors contributed equally to this work.}

\affil[1]{  \orgname{Secretária de Educação do Estado de São Paulo}, \orgaddress{\state{São Paulo}, \country{Brazil}}}

\affil*[2]{\orgdiv{Department of Mathematics, Institute of Geosciences and Exact Sciences}, \orgname{São Paulo State University (UNESP)},  \city{Rio Claro},  \state{SP}, \country{Brazil}}

\abstract{We investigate structural and enumerative properties of
binary single parity-check product codes. For each $n\geq 2$,
$\operatorname{SPC}(n)$ denotes the binary single parity-check code of
length $n$, consisting of all binary vectors of length $n$ having even
Hamming weight. We determine the generalized Hamming weight hierarchy
of the product code
$\mathcal{C}_{m,n}
=\operatorname{SPC}(m)\otimes\operatorname{SPC}(n)$,
whose codewords can be represented as $m\times n$ binary matrices in
which every row and every column has even Hamming weight. For the square
product
$\mathcal{C}_n
=\operatorname{SPC}(n)\otimes\operatorname{SPC}(n)$,
we also determine the maximum codeword weight and prove that its
homogeneous weight enumerator is symmetric if and only if $n$ is even.
After characterizing the dual code, we apply the MacWilliams identity
in its Walsh--Hadamard formulation to derive an exact closed-form
expression for the weight enumerator. By grouping the auxiliary binary
vectors according to their Hamming weights, we obtain an explicit
formula for each coefficient in terms of binomial coefficients and
alternating convolutions. Finally, using Krawtchouk polynomials, we
present an exact procedure for computing the full weight distribution
without exhaustively enumerating all codewords. Numerical examples
illustrate the formulas and verify the resulting computations.}



\maketitle

\section{Introduction}

The weight distribution is one of the fundamental invariants of a
linear code. It records the number of codewords of each Hamming weight
and, in particular, determines the minimum distance of the code. Weight
enumerators also play an important role in the study of duality,
decoding performance, and combinatorial identities associated with
linear codes (see, for example,
\cite{MacWilliamsSloane1977,Pless2003}). Despite their importance,
determining the full weight distribution of a linear code is generally
a difficult problem. Indeed, a direct computation for a binary linear
code of dimension $K$ may require examining all its $2^K$ codewords,
and the structure of a generator or parity-check matrix does not
usually lead immediately to a closed formula.

Product codes, introduced by Elias~\cite{Elias1954}, form a natural
family of linear codes obtained by imposing the constraints of shorter
component codes along different dimensions of an array. The binary
single parity-check code of length $n$, denoted by
$\operatorname{SPC}(n)$, consists of all vectors in
$\mathbb{F}_2^n$ having even Hamming weight. It has parameters
$[n,n-1,2]_2$. For integers $m,n\geq2$, we consider the product code
$
\mathcal{C}_{m,n}
=
\operatorname{SPC}(m)\otimes\operatorname{SPC}(n).
$
This is a binary linear code with parameters
$[mn,(m-1)(n-1),4]_2$, whose codewords can be represented as
$m\times n$ binary matrices in which every row and every column has
even Hamming weight. When $m=n$, we write
$\mathcal{C}_n=\mathcal{C}_{n,n}$.

Single parity-check product codes have been studied in connection with
iterative decoding and their performance over different communication
channels
\cite{rankin2001single,CaireTariccoBattail1995,Cardell2016b,Cardell2019b}.
Several works have also addressed weight enumeration for product codes.
Tolhuizen~\cite{tolhuizen2002more} characterized certain families of
low-weight codewords and expressed their multiplicities in terms of the
weight distributions of the component codes. El-Khamy and
Garello~\cite{ElKhamyGarello2006} observed that the full weight
enumerator remains unknown for most product codes and developed
average weight-enumerator approximations. The weight distribution and
asymptotic behavior of iterated single parity-check product codes were
studied by Caire, Taricco, and Battail
\cite{CaireTariccoBattail1995}. These results illustrate both the
importance of weight enumeration for product codes and the difficulty
of obtaining exact finite-length formulas.

Another invariant considered in this work is the generalized Hamming
weight hierarchy. Generalized Hamming weights, introduced by
Wei~\cite{Wei1991}, extend the notion of minimum distance by measuring
the smallest support of a subcode of a prescribed dimension. They
provide a finer description of the support structure of a linear code
and have applications, among others, to information security and
wiretap channels. For product codes whose components satisfy the chain
condition, a formula for the generalized Hamming weights was conjectured
by Wei and Yang~\cite{WeiYang1993} and proved by
Schaathun~\cite{Schaathun2000}; see also the extension obtained by
Martínez-Pérez and Willems~\cite{MartinezPerezWillems2004}. Since binary
single parity-check codes satisfy this condition, the formula provides
an explicit description of the complete generalized Hamming weight
hierarchy of their product codes.

Determining the coefficient $A_w$ in the weight enumerator of
$\mathcal{C}_{m,n}$ is equivalent to counting the $m\times n$ binary
matrices containing exactly $w$ entries equal to one and having even
row and column sums. Although each individual parity condition is
elementary, their simultaneous imposition makes the enumeration
nontrivial. Every matrix entry participates in both a row constraint
and a column constraint, and the complete family of $m+n$ parity
equations contains a global linear dependency. Moreover, direct
enumeration requires considering all
$2^{(m-1)(n-1)}$ codewords, which rapidly becomes impractical as the
component lengths increase.

We first apply the product formula for codes satisfying the chain
condition to determine the generalized Hamming weight hierarchy of
$\mathcal{C}_{m,n}$. For the square product $\mathcal{C}_n$, we also
determine the maximum possible Hamming weight of a codeword. It is
equal to $n^2$ when $n$ is even and to $n^2-n$ when $n$ is odd. As a
consequence, we prove that the homogeneous weight enumerator of
$\mathcal{C}_n$ is symmetric if and only if $n$ is even.

Our main enumerative result is an exact closed-form expression for the
weight enumerator of $\mathcal{C}_{m,n}$. We first characterize its
dual code and then apply the MacWilliams identity in its
Walsh--Hadamard formulation. The dual codewords are parametrized by two
auxiliary binary vectors, and their Hamming weights depend only on the
weights of these vectors. Grouping the corresponding terms according
to these two weights reduces the expression to a double sum and yields
explicit formulas for both the homogeneous and one-variable weight
enumerators. Expanding the resulting expression also provides a formula
for every coefficient $A_w$ in terms of binomial coefficients and
alternating convolutions.

The structure of the coefficient formula leads to an exact
computational procedure. Terms producing the same contribution are
first grouped together, after which the coefficients are evaluated
using a recurrence for binary Krawtchouk polynomials. For a code of
length $N=mn$, the resulting procedure requires at most $O(N^2)$
arithmetic operations and $O(N)$ auxiliary memory, while avoiding the
enumeration of the $2^{(m-1)(n-1)}$ codewords. Numerical examples for
different component lengths illustrate the formulas and provide
consistency checks for the computed weight distributions.

The remainder of the paper is organized as follows. We first recall the
necessary background on linear codes, single parity-check codes, product
codes, generalized Hamming weights, additive characters, the
Walsh--Hadamard transform, and the MacWilliams identity. We then study
the generalized Hamming weights, maximum weight, and symmetry properties
of single parity-check product codes. Next, we derive the closed-form
expression for the weight enumerator and its coefficients. The
computational procedure based on Krawtchouk polynomials is subsequently
presented together with numerical examples. Finally, we summarize the
main results and indicate possible directions for further research.

\section{Preliminaries}\label{sec:preliminaries}
\subsection{Linear Codes}

Let $q$ be a prime power and let $\mathbb{F}_q$ denote the finite
field with $q$ elements. A linear code $\mathcal{C}$ of length $n$
over $\mathbb{F}_q$ is a vector subspace of $\mathbb{F}_q^n$. If
$\dim_{\mathbb{F}_q}(\mathcal{C})=k$, then $\mathcal{C}$ is called an
$[n,k]_q$ linear code.

A generator matrix of $\mathcal{C}$ is a matrix
$G\in\mathbb{F}_q^{k\times n}$ whose rows form a basis of
$\mathcal{C}$. Thus,
\[
\mathcal{C}
=
\left\{
\mathbf{u}G:
\mathbf{u}\in\mathbb{F}_q^k
\right\},
\]
and every codeword $\mathbf{c}\in\mathcal{C}$ can be written uniquely
as
$
\mathbf{c}=\mathbf{u}G,
$
for some $\mathbf{u}\in\mathbb{F}_q^k$.

A parity-check matrix of $\mathcal{C}$ is a matrix
$H\in\mathbb{F}_q^{(n-k)\times n}$ of rank $n-k$ such that
\[
\mathcal{C}
=
\ker(H)
=
\left\{
\mathbf{c}\in\mathbb{F}_q^n:
H\mathbf{c}^{T}=\mathbf{0}
\right\}.
\]
Consequently, the generator and parity-check matrices satisfy $
GH^{T}=0.$
The parity-check matrix describes the parity constraints satisfied by
the codewords and plays a fundamental role in syndrome decoding \cite{Pless2003}.

The dual code of $\mathcal{C}$ is defined by
$$\mathcal{C}^{\perp}=\{\mathbf{x}\in\mathbb{F}_q^n:\mathbf{x}\cdot\mathbf{c}^t=0\text{ for every }\mathbf{c}\in\mathcal{C}\}.$$
The rows of $H$ form a basis of the dual code $\mathcal{C}^{\perp}$.

\begin{definition}
Let
$\mathbf{x}=(x_1,\ldots,x_n)\in\mathbb{F}_q^n$. The \emph{support} of
$\mathbf{x}$ is the set
\[
\operatorname{supp}(\mathbf{x})
=
\left\{
i\in\{1,\ldots,n\}:x_i\neq 0
\right\}.
\]
The \emph{Hamming weight} of $\mathbf{x}$, denoted by
$\operatorname{wt}(\mathbf{x})$, is the cardinality of its support:
\[
\operatorname{wt}(\mathbf{x})
=
\left|\operatorname{supp}(\mathbf{x})\right|.
\]
\end{definition}

\begin{definition}
The \emph{Hamming distance} between two vectors
$\mathbf{x},\mathbf{y}\in\mathbb{F}_q^n$ is defined by
\[
d_H(\mathbf{x},\mathbf{y})
=
\left|
\left\{
i\in\{1,\ldots,n\}:x_i\neq y_i
\right\}
\right|.
\]
Equivalently,
$
d_H(\mathbf{x},\mathbf{y})
=
\operatorname{wt}(\mathbf{x}-\mathbf{y}).
$
\end{definition}

The Hamming distance is translation invariant; that is,
$
d_H(\mathbf{x}+\mathbf{z},\mathbf{y}+\mathbf{z})
=
d_H(\mathbf{x},\mathbf{y})
$
for all $\mathbf{x},\mathbf{y},\mathbf{z}\in\mathbb{F}_q^n$. In
particular,
$
d_H(\mathbf{x},\mathbf{0})
=
\operatorname{wt}(\mathbf{x}).
$

\begin{definition}
The \emph{minimum Hamming distance} of a linear code $\mathcal{C}$ is
defined by
\[
d(\mathcal{C})
=
\min_{\substack{
\mathbf{x},\mathbf{y}\in\mathcal{C}\\
\mathbf{x}\neq\mathbf{y}
}}
d_H(\mathbf{x},\mathbf{y}).
\]
\end{definition}

Since $\mathcal{C}$ is linear and the Hamming distance is translation
invariant, its minimum distance can be determined from the nonzero
codewords:
\[
d(\mathcal{C})
=
\min\left\{
\operatorname{wt}(\mathbf{c}):
\mathbf{c}\in\mathcal{C}\setminus\{\mathbf{0}\}
\right\}.
\]
If $d(\mathcal{C})=d$, then $\mathcal{C}$ is called an
$[n,k,d]_q$ linear code.

\begin{definition}
Let $\mathcal{D}\leq\mathcal{C}$ be a linear subcode. The
\emph{support} of $\mathcal{D}$ is defined as
\[
\operatorname{supp}(\mathcal{D})
=
\bigcup_{\mathbf{c}\in\mathcal{D}}
\operatorname{supp}(\mathbf{c}).
\]
Equivalently,
\[
\operatorname{supp}(\mathcal{D})
=
\left\{
i\in\{1,\ldots,n\}:
\text{there exists }\mathbf{c}\in\mathcal{D}
\text{ such that }c_i\neq 0
\right\}.
\]
\end{definition}

The notion of minimum Hamming distance can be generalized by
considering subcodes of prescribed dimensions. Generalized Hamming
weights were introduced by Wei~\cite{Wei1991}.

\begin{definition}
Let $\mathcal{C}$ be an $[n,k,d]_q$ linear code. For
$1\leq r\leq k$, the \emph{$r$-th generalized Hamming weight} of
$\mathcal{C}$ is defined by
\[
d_r(\mathcal{C})
=
\min\left\{
\left|\operatorname{supp}(\mathcal{D})\right|:
\mathcal{D}\leq\mathcal{C},
\ \dim_{\mathbb{F}_q}(\mathcal{D})=r
\right\}.
\]
\end{definition}

The sequence
\[
\left(
d_1(\mathcal{C}),
d_2(\mathcal{C}),
\ldots,
d_k(\mathcal{C})
\right)
\]
is called the \emph{weight hierarchy}, or the
\emph{generalized Hamming weight hierarchy}, of $\mathcal{C}$. The
first generalized Hamming weight coincides with the minimum Hamming
distance, $
d_1(\mathcal{C})=d(\mathcal{C}).
$
Moreover, the generalized Hamming weights satisfy
\[
d_1(\mathcal{C})
<
d_2(\mathcal{C})
<
\cdots
<
d_k(\mathcal{C}),
\]
and
\[
d_k(\mathcal{C})
=
\left|\operatorname{supp}(\mathcal{C})\right|
\leq n.
\]

\begin{definition}
Let $\mathcal{C}$ be an $[n,k,d]_q$ linear code. For each
$0\leq w\leq n$, let
\[
A_w
=
\left|
\left\{
\mathbf{c}\in\mathcal{C}:
\operatorname{wt}(\mathbf{c})=w
\right\}
\right|.
\]
The sequence $(A_0,A_1,\ldots,A_n)$ is called the
\emph{weight distribution} of $\mathcal{C}$. The \emph{homogeneous
weight enumerator} of $\mathcal{C}$ is the polynomial
\[
W_{\mathcal{C}}(x,y)
=
\sum_{w=0}^{n}
A_w x^{n-w}y^w.
\]
\end{definition}

The homogeneous weight enumerator can also be expressed in one-variable form by setting $x=1$ and $y=z$:
$$
W_{\mathcal{C}}(z)
=
W_{\mathcal{C}}(1,z)
=
\sum_{w=0}^{n}A_wz^w.
$$


\subsection{Single Parity-Check Codes}
 
Let $n\geq 2$. The binary \emph{Single Parity-Check Code} of length
$n$, denoted by $\operatorname{SPC}(n)$, is defined as
\[
\operatorname{SPC}(n)
=
\left\{
\mathbf{x}\in\mathbb{F}_2^n:
\sum_{i=1}^{n}x_i=0
\right\}.
\]
Since the sum is computed over $\mathbb{F}_2$, a vector belongs to
$\operatorname{SPC}(n)$ if and only if it has even Hamming weight.

A parity-check matrix of $\operatorname{SPC}(n)$ is
\[
H
=
\begin{bmatrix}
1 & 1 & \cdots & 1
\end{bmatrix}.
\]
Since $\operatorname{rank}(H)=1$, the code has dimension $n-1$. A
generator matrix is given by
\[
G
=
\left[
I_{n-1}
\mid
\mathbf{1}_{n-1}
\right],
\]
where $I_{n-1}$ is the identity matrix of order $n-1$ and
$\mathbf{1}_{n-1}$ is the all-ones column vector of length $n-1$.

Every nonzero codeword of $\operatorname{SPC}(n)$ has even Hamming
weight and therefore has weight at least two. Moreover, the code
contains codewords of weight two. Hence,
the minimum distance of the code  is $d=2$,
and $\operatorname{SPC}(n)$ is an $[n,n-1,2]_2$ linear code over~$\mathbb{F}_2$.


\begin{theorem}
The weight enumerator polynomial of the code $\operatorname{SPC}(n)$
is given by
\[
W_{\operatorname{SPC}(n)}(x,y)
=
\sum_{i=0}^{\lfloor n/2\rfloor}
\binom{n}{2i}x^{n-2i}y^{2i}.
\]
\end{theorem}

\begin{proof}
The result follows from the fact that $\operatorname{SPC}(n)$ consists
of all binary words of length $n$ and even Hamming weight. For each
$i$ satisfying
\[
0\leq i\leq\left\lfloor\frac{n}{2}\right\rfloor,
\]
the number of codewords of weight $2i$ is
\[
\binom{n}{2i}.
\]
Since each codeword $\mathbf{c}$ contributes the term
\[
x^{n-\operatorname{wt}(\mathbf{c})}y^{\operatorname{wt}(\mathbf{c})}
\]
to the weight enumerator, the stated expression follows.
\end{proof}




\subsection{ SPC product codes}

Product codes were introduced by Elias~\cite{Elias1954} as a method
for constructing long block codes from shorter component codes by
imposing coding constraints along different dimensions of a
multidimensional array.

Let $C_1$ and $C_2$ be binary linear codes with parameters
$[n_1,k_1,d_1]_2$ and $[n_2,k_2,d_2]_2$, and let $G_1$ and $G_2$ be
generator matrices of $C_1$ and $C_2$, respectively. The
\emph{product code} of $C_1$ and $C_2$, denoted by
$C_1\otimes C_2$, is the linear code generated by
\[
G_1\otimes G_2.
\]
It has parameters
$
[n_1n_2,k_1k_2,d_1d_2]_2.
$

After arranging its coordinates row by row, a codeword of
$C_1\otimes C_2$ can be represented as an
$n_1\times n_2$ binary matrix
$
X=[x_{ij}].
$
Under this convention, every row of $X$ is a codeword of $C_2$, and
every column of $X$ is a codeword of $C_1$. More precisely, every
codeword matrix can be written as
\[
X=G_1^{T}UG_2,
\]
for some $U\in\mathbb{F}_2^{k_1\times k_2}$.

Let $H_1$ and $H_2$ be parity-check matrices of $C_1$ and $C_2$,
respectively. Then $X$ represents a codeword of $C_1\otimes C_2$ if
and only if
\[
H_1X=0
\qquad\text{and}\qquad
XH_2^{T}=0.
\]
The first equation imposes the parity-check constraints of $C_1$ on
every column of $X$, whereas the second imposes those of $C_2$ on
every row.

With respect to the row-wise ordering of the entries of $X$, these
constraints are represented by the possibly redundant parity-check
matrix
\[
\mathcal{H}(C_1,C_2)
=
\begin{bmatrix}
I_{n_1}\otimes H_2\\
H_1\otimes I_{n_2}
\end{bmatrix}.
\]
Its rank is $
\operatorname{rank}\bigl(\mathcal{H}(C_1,C_2)\bigr)
=
n_1n_2-k_1k_2.
$
Consequently, the number of redundant rows in
$\mathcal{H}(C_1,C_2)$ is
$
(n_1-k_1)(n_2-k_2).
$

We now specialize this construction to two single parity-check codes.
Let $n,m\geq 2$ and consider
\[
\mathcal{C}_{n,m}
=
\operatorname{SPC}(n)\otimes\operatorname{SPC}(m).
\]
Since the component codes have parameters
$
[n,n-1,2]_2
$ and $
[m,m-1,2]_2,
$
respectively, the product code has parameters
$
[nm,(n-1)(m-1),4]_2.
$

A codeword of $\mathcal{C}_{n,m}$ is represented as an
$n\times m$ binary matrix
$
X=[x_{ij}],
$
such that every row and every column has even Hamming weight.
Equivalently,
\[
\sum_{j=1}^{m}x_{ij}=0,
\
1\leq i\leq n,
\quad \text{and} \quad 
\sum_{i=1}^{n}x_{ij}=0,
\
1\leq j\leq m.
\]

Let
\[
H_n
=
\begin{bmatrix}
1&1&\cdots&1
\end{bmatrix}
\in\mathbb{F}_2^{1\times n},
\]
\[
H_m
=
\begin{bmatrix}
1&1&\cdots&1
\end{bmatrix}
\in\mathbb{F}_2^{1\times m},
\]
be parity-check matrices of $\operatorname{SPC}(n)$ and
$\operatorname{SPC}(m)$, respectively. The row and column parity
constraints of $\mathcal{C}_{n,m}$ are represented by
\[
\mathcal{H}_{n,m}
=
\begin{bmatrix}
I_n\otimes H_m\\
H_n\otimes I_m
\end{bmatrix}
=
\left[
\begin{array}{ccccc}
\mathbf{1}_{1\times m}
    &0&\cdots&0&0\\
0&\mathbf{1}_{1\times m}
    &\cdots&0&0\\
\vdots&\vdots&\ddots&\vdots&\vdots\\
0&0&\cdots&0&\mathbf{1}_{1\times m}\\
\hline
I_m&I_m&\cdots&I_m&I_m
\end{array}
\right],
\]
where the upper part contains $n$ rows and
$\mathbf{1}_{1\times m}$ denotes the all-ones row vector of length
$m$. The first $n$ rows impose the parity constraints on the rows of
$X$, while the last $m$ rows impose the parity constraints on its
columns.

The matrix $\mathcal{H}_{n,m}$ has size
$
(n+m)\times nm.
$
Its rows satisfy exactly one linear dependency. Indeed, over
$\mathbb{F}_2$, the sum of all row-parity constraints is equal to the
sum of all column-parity constraints. Therefore,
\[
\operatorname{rank}(\mathcal{H}_{n,m})
=
n+m-1.
\]
Thus, a full-rank parity-check matrix for $\mathcal{C}_{n,m}$ can be
obtained by removing any one row from $\mathcal{H}_{n,m}$.


\begin{example}\label{ex:G4}
For $n=m=4$, the code
\[
\mathcal{C}_{4}
=
\operatorname{SPC}(4)\otimes\operatorname{SPC}(4)
\]
has parameters
$
[16,9,4]_2.
$
Its redundant parity-check matrix is
\[
\mathcal{H}_{4,4}
=
\left[
\begin{array}{cccc|cccc|cccc|cccc}
1&1&1&1&0&0&0&0&0&0&0&0&0&0&0&0\\
0&0&0&0&1&1&1&1&0&0&0&0&0&0&0&0\\
0&0&0&0&0&0&0&0&1&1&1&1&0&0&0&0\\
0&0&0&0&0&0&0&0&0&0&0&0&1&1&1&1\\\hline
1&0&0&0&1&0&0&0&1&0&0&0&1&0&0&0\\
0&1&0&0&0&1&0&0&0&1&0&0&0&1&0&0\\
0&0&1&0&0&0&1&0&0&0&1&0&0&0&1&0\\
0&0&0&1&0&0&0&1&0&0&0&1&0&0&0&1
\end{array}
\right].
\]
The first four rows impose the parity constraints on the rows of the
codeword matrix, while the last four impose those on its columns.
Since these eight parity-check equations contain one linear
dependency,
\[
\operatorname{rank}(\mathcal{H}_{4,4})=7.
\]
Consequently, removing any one row from $\mathcal{H}_{4,4}$ produces
a full-rank $7\times16$ parity-check matrix.
Moreover, a generator matrix of $\mathcal{C}_{4}$ is
\[
G_{4,4}
=
G_{\operatorname{SPC}(4)}
\otimes
G_{\operatorname{SPC}(4)}
=
\begin{bmatrix}
1&0&0&1\\
0&1&0&1\\
0&0&1&1
\end{bmatrix}
\otimes 
\begin{bmatrix}
1&0&0&1\\
0&1&0&1\\
0&0&1&1
\end{bmatrix}=
\left[
\begin{array}{cccc|cccc|cccc|cccc}
1&0&0&1&0&0&0&0&0&0&0&0&1&0&0&1\\
0&1&0&1&0&0&0&0&0&0&0&0&0&1&0&1\\
0&0&1&1&0&0&0&0&0&0&0&0&0&0&1&1\\\hline
0&0&0&0&1&0&0&1&0&0&0&0&1&0&0&1\\
0&0&0&0&0&1&0&1&0&0&0&0&0&1&0&1\\
0&0&0&0&0&0&1&1&0&0&0&0&0&0&1&1\\\hline
0&0&0&0&0&0&0&0&1&0&0&1&1&0&0&1\\
0&0&0&0&0&0&0&0&0&1&0&1&0&1&0&1\\
0&0&0&0&0&0&0&0&0&0&1&1&0&0&1&1
\end{array}
\right].
\]
\end{example}

\subsection{Characters and the Walsh--Hadamard Transform}

Let $\mathbb{G}$ be a finite additive abelian group. A
\emph{character} of $\mathbb{G}$ is a group homomorphism
\[
\chi\colon\mathbb{G}\longrightarrow\mathbb{C}^{\times},
\]
where
$
\mathbb{C}^{\times}
=
\mathbb{C}\setminus\{0\}
$
is the multiplicative group of nonzero complex numbers. Thus,
\[
\chi(x+y)=\chi(x)\chi(y),
\]
for every $x,y\in\mathbb{G}$. The set of all characters of
$\mathbb{G}$, denoted by $\widehat{\mathbb{G}}$, forms a group under
pointwise multiplication and is called the \emph{dual group} of
$\mathbb{G}$.

Consider now the additive group $\mathbb{F}_2^n$. For each
$\mathbf{a}\in\mathbb{F}_2^n$, define
\[
\chi_{\mathbf{a}}(\mathbf{x})
=
(-1)^{\mathbf{a}\cdot\mathbf{x}},
\qquad
\mathbf{x}\in\mathbb{F}_2^n,
\]
where
\[
\mathbf{a}\cdot\mathbf{x}
=
\sum_{i=1}^{n}a_i x_i
\in\mathbb{F}_2
\]
is the standard inner product over $\mathbb{F}_2$.

Every character of $\mathbb{F}_2^n$ is of this form. Indeed, since
$2\mathbf{x}=\mathbf{0}$ for every $\mathbf{x}\in\mathbb{F}_2^n$, any
character $\chi$ satisfies
\[
\chi(\mathbf{x})^2
=
\chi(2\mathbf{x})
=
\chi(\mathbf{0})
=
1.
\]
Therefore, every character takes values in $\{1,-1\}$. Moreover, the
map
\[
\mathbf{a}\longmapsto\chi_{\mathbf{a}}
\]
defines an isomorphism between $\mathbb{F}_2^n$ and its dual group
$\widehat{\mathbb{F}_2^n}$.

The characters satisfy the orthogonality relation
\[
\sum_{\mathbf{x}\in\mathbb{F}_2^n}
\chi_{\mathbf{a}}(\mathbf{x})
\overline{\chi_{\mathbf{b}}(\mathbf{x})}
=
\sum_{\mathbf{x}\in\mathbb{F}_2^n}
(-1)^{(\mathbf{a}+\mathbf{b})\cdot\mathbf{x}}
=
\begin{cases}
2^n, & \mathbf{a}=\mathbf{b},\\
0,   & \mathbf{a}\neq\mathbf{b}.
\end{cases}
\]
Notice that the complex conjugate can be omitted in this case, since
the characters take only real values. For further details, see, for
example, \cite{rudin2017fourier,odonnell2014analysis}.

\begin{definition}[Walsh--Hadamard transform]
Let $f\colon\mathbb{F}_2^n\to\mathbb{C}$. The
\emph{Walsh--Hadamard transform} of $f$ is the function
$\widehat{f}\colon\mathbb{F}_2^n\to\mathbb{C}$ defined by
\[
\widehat{f}(\mathbf{a})
=
\sum_{\mathbf{x}\in\mathbb{F}_2^n}
f(\mathbf{x})(-1)^{\mathbf{a}\cdot\mathbf{x}}.
\]
\end{definition}

This is the unnormalized Fourier transform on the additive group
$\mathbb{F}_2^n$. By the orthogonality of the characters, its
inversion formula is
\[
f(\mathbf{x})
=
\frac{1}{2^n}
\sum_{\mathbf{a}\in\mathbb{F}_2^n}
\widehat{f}(\mathbf{a})
(-1)^{\mathbf{a}\cdot\mathbf{x}}.
\]

Let $\mathcal{C}$  be a binary linear code in $\mathbb{F}_2^n$. Its dual
code is defined by
\[
\mathcal{C}^{\perp}
=
\left\{
\mathbf{a}\in\mathbb{F}_2^n:
\mathbf{a}\cdot\mathbf{x}=0
\text{ for every }\mathbf{x}\in\mathcal{C}
\right\}.
\]
We denote by $\mathbf{1}_{\mathcal{C}}$ the indicator function of
$\mathcal{C}$, defined by
\[
\mathbf{1}_{\mathcal{C}}(\mathbf{x})
=
\begin{cases}
1, & \mathbf{x}\in\mathcal{C},\\
0, & \mathbf{x}\notin\mathcal{C}.
\end{cases}
\]

The following standard result describes the Walsh--Hadamard transform
of the indicator function of a linear code.

\begin{proposition}\label{prop:WH-indicator-code}
Let $\mathcal{C}$ be a binary linear code in $\mathbb{F}_2^n$. Then
\[
\widehat{\mathbf{1}_{\mathcal{C}}}(\mathbf{a})
=
\sum_{\mathbf{x}\in\mathcal{C}}
(-1)^{\mathbf{a}\cdot\mathbf{x}}
=
\begin{cases}
|\mathcal{C}|, & \mathbf{a}\in\mathcal{C}^{\perp},\\
0,             & \mathbf{a}\notin\mathcal{C}^{\perp}.
\end{cases}
\]
Equivalently,
$
\widehat{\mathbf{1}_{\mathcal{C}}}
=
|\mathcal{C}|\mathbf{1}_{\mathcal{C}^{\perp}}.
$
\end{proposition}

\begin{proof}
Suppose first that $\mathbf{a}\in\mathcal{C}^{\perp}$. Then
$
\mathbf{a}\cdot\mathbf{x}=0
$
for every $\mathbf{x}\in\mathcal{C}$, and hence
\[
\widehat{\mathbf{1}_{\mathcal{C}}}(\mathbf{a})
=
\sum_{\mathbf{x}\in\mathcal{C}}1
=
|\mathcal{C}|.
\]

Now suppose that $\mathbf{a}\notin\mathcal{C}^{\perp}$. Then there exists
$\mathbf{x}_0\in\mathcal{C}$ such that
$
\mathbf{a}\cdot\mathbf{x}_0=1.
$
Set
\[
S
=
\sum_{\mathbf{x}\in\mathcal{C}}
(-1)^{\mathbf{a}\cdot\mathbf{x}}.
\]
Since $\mathcal{C}$ is linear, the map
\[
\mathbf{x}\longmapsto\mathbf{x}+\mathbf{x}_0
\]
is a bijection of $\mathcal{C}$. Therefore,
\[
S
=
\sum_{\mathbf{x}\in\mathcal{C}}
(-1)^{\mathbf{a}\cdot(\mathbf{x}+\mathbf{x}_0)}\\
=
(-1)^{\mathbf{a}\cdot\mathbf{x}_0}
\sum_{\mathbf{x}\in\mathcal{C}}
(-1)^{\mathbf{a}\cdot\mathbf{x}}\\
=
-S.
\]
It follows that $S=0$, which proves the result.
\end{proof}

\subsection{The MacWilliams Identity via the Walsh--Hadamard Transform}

The Walsh--Hadamard transform provides a natural Fourier-analytic
formulation of the MacWilliams identity. For a binary linear code
$\mathcal{C}$ in $\mathbb{F}_2^n$, this identity relates the weight
enumerator of $\mathcal{C}$ to that of its dual code
$\mathcal{C}^{\perp}$. We derive it from the orthogonality of the
characters of $\mathbb{F}_2^n$.

We first recall the discrete Poisson summation formula.

\begin{theorem}[Poisson summation formula]
\label{thm:poisson-f2}
Let $\mathcal{C}$ be a binary linear code in $\mathbb{F}_2^n$  and let
$f\colon\mathbb{F}_2^n\to\mathbb{C}$. Then
\[
\sum_{\mathbf{v}\in\mathcal{C}^{\perp}}f(\mathbf{v})
=
\frac{1}{|\mathcal{C}|}
\sum_{\mathbf{u}\in\mathcal{C}}
\widehat{f}(\mathbf{u}).
\]
\end{theorem}

\begin{proof}
By the Walsh--Hadamard inversion formula,
\[
f(\mathbf{v})
=
\frac{1}{2^n}
\sum_{\mathbf{a}\in\mathbb{F}_2^n}
\widehat{f}(\mathbf{a})
(-1)^{\mathbf{a}\cdot\mathbf{v}}.
\]
Summing over $\mathbf{v}\in\mathcal{C}^{\perp}$ gives
\[
\begin{aligned}
\sum_{\mathbf{v}\in\mathcal{C}^{\perp}}f(\mathbf{v})
&=
\frac{1}{2^n}
\sum_{\mathbf{a}\in\mathbb{F}_2^n}
\widehat{f}(\mathbf{a})
\sum_{\mathbf{v}\in\mathcal{C}^{\perp}}
(-1)^{\mathbf{a}\cdot\mathbf{v}}.
\end{aligned}
\]
Applying Proposition~\ref{prop:WH-indicator-code} to
$\mathcal{C}^{\perp}$ and using
$
\left(\mathcal{C}^{\perp}\right)^{\perp}
=
\mathcal{C},
$
we obtain
\[
\sum_{\mathbf{v}\in\mathcal{C}^{\perp}}
(-1)^{\mathbf{a}\cdot\mathbf{v}}
=
\begin{cases}
|\mathcal{C}^{\perp}|,
    & \mathbf{a}\in\mathcal{C},\\
0,
    & \mathbf{a}\notin\mathcal{C}.
\end{cases}
\]
Consequently,
\[
\sum_{\mathbf{v}\in\mathcal{C}^{\perp}}f(\mathbf{v})
=
\frac{|\mathcal{C}^{\perp}|}{2^n}
\sum_{\mathbf{u}\in\mathcal{C}}
\widehat{f}(\mathbf{u}).
\]
Since
$
|\mathcal{C}|\,|\mathcal{C}^{\perp}|=2^n,
$
we have
\[
\frac{|\mathcal{C}^{\perp}|}{2^n}
=
\frac{1}{|\mathcal{C}|},
\]
which proves the result.
\end{proof}

We now apply the Poisson summation formula to derive the MacWilliams
identity. Recall that the homogeneous weight enumerator of
$\mathcal{C}$ is
\[
W_{\mathcal{C}}(x,y)
=
\sum_{\mathbf{v}\in\mathcal{C}}
x^{n-\operatorname{wt}(\mathbf{v})}y^{\operatorname{wt}(\mathbf{v})}.
\]

\begin{theorem}[MacWilliams identity]
\label{thm:macwilliams-binary}
Let $\mathcal{C}$ be a binary linear code in $\mathbb{F}_2^n$. Then
\[
W_{\mathcal{C}^{\perp}}(x,y)
=
\frac{1}{|\mathcal{C}|}
W_{\mathcal{C}}(x+y,x-y).
\]
\end{theorem}

\begin{proof}
Consider the polynomial-valued function
\[
f\colon\mathbb{F}_2^n\longrightarrow\mathbb{C}[x,y]
\]
defined by
\[
f(\mathbf{v})
=
x^{n-\operatorname{wt}(\mathbf{v})}y^{\operatorname{wt}(\mathbf{v})}.
\]
The Poisson summation formula applies coefficientwise to such
functions.

For $\mathbf{u}\in\mathbb{F}_2^n$, the Walsh--Hadamard transform of $f$
can be computed coordinatewise:
\[
\begin{aligned}
\widehat{f}(\mathbf{u})
&=
\sum_{\mathbf{v}\in\mathbb{F}_2^n}
x^{n-\operatorname{wt}(\mathbf{v})}
y^{\operatorname{wt}(\mathbf{v})}
(-1)^{\mathbf{u}\cdot\mathbf{v}}\\
&=
\prod_{i=1}^{n}
\left(
\sum_{v_i\in\{0,1\}}
x^{1-v_i}y^{v_i}(-1)^{u_i v_i}
\right).
\end{aligned}
\]
For each coordinate,
\[
\sum_{v_i\in\{0,1\}}
x^{1-v_i}y^{v_i}(-1)^{u_i v_i}
=
\begin{cases}
x+y, & u_i=0,\\
x-y, & u_i=1.
\end{cases}
\]
Therefore,
\[
\widehat{f}(\mathbf{u})
=
(x+y)^{n-\operatorname{wt}(\mathbf{u})}
(x-y)^{\operatorname{wt}(\mathbf{u})}.
\]

Applying Theorem~\ref{thm:poisson-f2}, we obtain
\[
\begin{aligned}
W_{\mathcal{C}^{\perp}}(x,y)
&=
\sum_{\mathbf{v}\in\mathcal{C}^{\perp}}
x^{n-\operatorname{wt}(\mathbf{v})}
y^{\operatorname{wt}(\mathbf{v})}\\
&=
\frac{1}{|\mathcal{C}|}
\sum_{\mathbf{u}\in\mathcal{C}}
\widehat{f}(\mathbf{u})\\
&=
\frac{1}{|\mathcal{C}|}
\sum_{\mathbf{u}\in\mathcal{C}}
(x+y)^{n-\operatorname{wt}(\mathbf{u})}
(x-y)^{\operatorname{wt}(\mathbf{u})}\\
&=
\frac{1}{|\mathcal{C}|}
W_{\mathcal{C}}(x+y,x-y).
\end{aligned}
\]
\end{proof}

Equivalently, if
\[
W_{\mathcal{C}}(x,y)
=
\sum_{w=0}^{n}
A_w x^{n-w}y^w,
\]
then
\[
W_{\mathcal{C}^{\perp}}(x,y)
=
\frac{1}{|\mathcal{C}|}
\sum_{w=0}^{n}
A_w(x+y)^{n-w}(x-y)^w.
\]

\section{Weight Properties of SPC Product Codes}
\label{sec:weight-properties-spc}

The notion of generalized Hamming weights was introduced by
Wei~\cite{Wei1991} as a refinement of the minimum Hamming distance of a
linear code. While the minimum distance is the smallest support size of
a nonzero codeword, the $r$-th generalized Hamming weight is the
smallest support size of an $r$-dimensional subcode. The resulting
weight hierarchy therefore provides a finer description of the support
structure of a linear code.

Wei introduced generalized Hamming weights in connection with the
wiretap channel of type II, where they characterize the amount of
information that may be revealed through a partially observed
transmission. This interpretation motivates the study of the entire
weight hierarchy rather than only the minimum distance. In particular,
the first generalized Hamming weight coincides with the minimum
distance, whereas the higher generalized Hamming weights describe the
minimum number of coordinates required to support subcodes of
increasing dimensions.

A natural question is how generalized Hamming weights behave under
standard code constructions, such as products of linear codes. Wei and
Yang~\cite{WeiYang1993} conjectured a formula for the generalized
Hamming weights of a product code in terms of those of its component
codes, provided that both components satisfy the chain condition. The
conjecture was proved by Schaathun~\cite{Schaathun2000}, and subsequent
work extended the characterization to products with more than two
components (see, for example,~\cite{MartinezPerezWillems2004}).

SPC codes satisfy the chain condition, since their
generalized Hamming weights can be attained by a nested sequence of
subcodes. In what follows, we give a self-contained derivation for
products binary SPC codes. Although our main object
of study is the   product
$
\mathcal{C}_n
=
\operatorname{SPC}(n)\otimes\operatorname{SPC}(n),
$
it is convenient to first consider the rectangular product
$\operatorname{SPC}(m)\otimes\operatorname{SPC}(n)$. The parameters
$m$ and $n$ denote the lengths of the component codes.

\subsection{The code $\operatorname{SPC}(2)\otimes
\operatorname{SPC}(n)$}

\begin{proposition}
Let $n\geq 2$ and let
$
\mathcal C_{2,n}
=
\operatorname{SPC}(2)\otimes\operatorname{SPC}(n)
$
be the binary linear code generated by
\[
G
=
\begin{bmatrix}1&1\end{bmatrix}
\otimes
\begin{bmatrix}I_{n-1}&\mathbf{1}_{n-1}\end{bmatrix}.
\]
Then $\mathcal{C}_{2,n}$ has parameters $[2n,n-1,4]_2$, and its
generalized Hamming weights are
\[
d_r(\mathcal{C}_{2,n})=2(r+1),
\qquad
1\leq r\leq n-1.
\]
Equivalently, its generalized Hamming weight hierarchy is
\[
\bigl(d_1(\mathcal{C}_{2,n}),\ldots,d_{n-1}(\mathcal{C}_{2,n})\bigr)
=
(4,6,8,\ldots,2n).
\]
\end{proposition}

\begin{proof}
Let
$
G_n
=
\begin{bmatrix}I_{n-1}&\mathbf{1}_{n-1}\end{bmatrix}
$
be a generator matrix of $\operatorname{SPC}(n)$. Since
\[
G
=
\begin{bmatrix}G_n&G_n\end{bmatrix},
\]
the code $\mathcal{C}_{2,n}$ is the image of $\operatorname{SPC}(n)$ under
the injective linear map
\[
\varphi\colon\operatorname{SPC}(n)\longrightarrow\mathcal{C},
\qquad
\varphi(\mathbf{x})=(\mathbf{x},\mathbf{x}).
\]
Consequently, $\mathcal{C}_{2,n}$ has dimension $n-1$, and every support is
duplicated under $\varphi$.

We first observe that
\[
d_r\bigl(\operatorname{SPC}(n)\bigr)=r+1,
\qquad
1\leq r\leq n-1.
\]
Indeed, if an $r$-dimensional subcode is supported on a set of $s$
coordinates, then it is contained in the single parity-check code on
those $s$ coordinates, which has dimension $s-1$. Hence
$
r\leq s-1,
$
and therefore $s\geq r+1$. Conversely, the set of all even-weight
vectors supported on any fixed set of $r+1$ coordinates is an
$r$-dimensional subcode of $\operatorname{SPC}(n)$. This proves the
claim.

Every $r$-dimensional subcode of $\mathcal{C}_{2,n}$ is the image under
$\varphi$ of an $r$-dimensional subcode of $\operatorname{SPC}(n)$.
Since $\varphi$ duplicates the support, i.e., $|\operatorname{supp}(\varphi(\mathcal D))|=2|\operatorname{supp}(\mathcal D)|$, it follows that
\[
d_r(\mathcal{C})
=
2d_r\bigl(\operatorname{SPC}(n)\bigr)
=
2(r+1).
\]
\end{proof}
 
\subsection{The code
$\operatorname{SPC}(m)\otimes\operatorname{SPC}(n)$}

Let $m,n\geq2$ and consider the binary product code
$\mathcal{C}_{m,n}
=\operatorname{SPC}(m)\otimes\operatorname{SPC}(n)$.
This code has parameters $[mn,(m-1)(n-1),4]_2$.

We first recall the chain condition for generalized Hamming weights \cite{Encheva1994}. An
$[n,k]_q$ linear code $\mathcal{C}$ is said to satisfy the
\emph{chain condition} if there exists a sequence of subcodes
$\{0\}=\mathcal{D}_0\subset\mathcal{D}_1\subset\cdots
\subset\mathcal{D}_k=\mathcal{C}$ such that
$\dim(\mathcal{D}_i)=i$ and
$|\operatorname{supp}(\mathcal{D}_i)|=d_i(\mathcal{C})$ for every
$1\leq i\leq k$.

For $\ell\geq2$, the code $\operatorname{SPC}(\ell)$ satisfies the
chain condition. Indeed, for $1\leq i\leq\ell-1$, let
$\mathcal{D}_i^{(\ell)}$ be the subcode consisting of all even-weight
vectors supported on the first $i+1$ coordinates. Then
$$\dim(\mathcal{D}_i^{(\ell)})=i \quad \text{and} \quad |\operatorname{supp}(\mathcal{D}_i^{(\ell)})|=i+1.$$ 
Since
$d_i(\operatorname{SPC}(\ell))=i+1$, the nested sequence
$$\mathcal{D}_1^{(\ell)}\subset\cdots\subset
\mathcal{D}_{\ell-1}^{(\ell)}=\operatorname{SPC}(\ell)$$
realizes the complete weight hierarchy.

The generalized Hamming weights of product codes satisfying the chain
condition can be determined from the weight hierarchies of their
component codes. The following formula was conjectured by Wei and
Yang~\cite{WeiYang1993} and proved by
Schaathun~\cite{Schaathun2000}.

\begin{theorem}[Product formula for chained codes]
\label{thm:product-chain-formula}
Let $\mathcal{A}$ and $\mathcal{B}$ be linear codes of dimensions
$k_{\mathcal{A}}$ and $k_{\mathcal{B}}$, respectively, and suppose
that both codes satisfy the chain condition. Set
$d_0(\mathcal{A})=0$. Then, for
$1\leq r\leq k_{\mathcal{A}}k_{\mathcal{B}}$,
\[
d_r(\mathcal{A}\otimes\mathcal{B})
=
\min_{\substack{
1\leq s\leq k_{\mathcal{A}}\\
k_{\mathcal{B}}\geq t_1\geq\cdots\geq t_s\geq1\\
t_1+\cdots+t_s=r
}}
\sum_{i=1}^{s}
\bigl(d_i(\mathcal{A})-d_{i-1}(\mathcal{A})\bigr)
d_{t_i}(\mathcal{B}).
\]
\end{theorem}

Now, we introduce the weight hierarchy of $\mathcal{C}_{m,n}$.

\begin{theorem}\label{thm:ghw-spc-product}
Let $m,n\geq2$. For every
$1\leq r\leq(m-1)(n-1)$, the $r$-th generalized Hamming weight of
$\mathcal{C}_{m,n}$ is
\[
d_r(\mathcal{C}_{m,n})
=
r+1+\lambda_{m,n}(r),
\]
where
\[
\lambda_{m,n}(r)
=
\min\left\{
a+b:
1\leq a\leq m-1,\ 
1\leq b\leq n-1,\ 
ab\geq r
\right\}.
\]
Equivalently,
\[
d_r(\mathcal{C}_{m,n})
=
r+1+
\min_{\substack{
1\leq a\leq m-1\\
1\leq b\leq n-1\\
ab\geq r
}}
(a+b).
\]
\end{theorem}

\begin{proof}
Set
$\mathcal{A}=\operatorname{SPC}(m)$ and
$\mathcal{B}=\operatorname{SPC}(n)$. Their generalized Hamming weights
are
$d_i(\mathcal{A})=i+1$ for $1\leq i\leq m-1$ and
$d_j(\mathcal{B})=j+1$ for $1\leq j\leq n-1$.

Since $d_0(\mathcal{A})=0$, we have
\[
d_i(\mathcal{A})-d_{i-1}(\mathcal{A})
=
\begin{cases}
2, & i=1,\\
1, & 2\leq i\leq m-1.
\end{cases}
\]
Applying Theorem~\ref{thm:product-chain-formula}, we obtain
\[
\begin{aligned}
d_r(\mathcal{C}_{m,n})
&=
\min
\left\{
2(t_1+1)
+
\sum_{i=2}^{s}(t_i+1)
\right\}&=
\min
\left\{
2t_1+2+
\sum_{i=2}^{s}t_i+s-1
\right\}\\
&=
r+1+\min\left\{s+t_1\right\},\\[5pt]
\end{aligned}
\]
where each minimum is taken over all integers $s,t_1,\ldots,t_s$
satisfying
$1\leq s\leq m-1$,
$n-1\geq t_1\geq\cdots\geq t_s\geq1$, and
$t_1+\cdots+t_s=r$.

For a fixed value of $s$, the condition
$t_1+\cdots+t_s=r$ implies
$t_1\geq\lceil r/s\rceil$. Conversely, this bound is attained by
distributing $r$ as evenly as possible among the $s$ positive
integers. More precisely, writing $r=qs+\rho$, where
$0\leq\rho<s$, we may take
\[
t_1=\cdots=t_{\rho}=q+1
\qquad\text{and}\qquad
t_{\rho+1}=\cdots=t_s=q.
\]
Therefore,
\[
\min\left\{s+t_1\right\}
=
\min_{\substack{
1\leq s\leq\min\{r,m-1\}\\
\lceil r/s\rceil\leq n-1
}}
\left(
s+\left\lceil\frac{r}{s}\right\rceil
\right).
\]

On the other hand, for a fixed $a$, the smallest integer $b$ satisfying
$ab\geq r$ is $b=\lceil r/a\rceil$. Hence,
\[
\lambda_{m,n}(r)
=
\min_{\substack{
1\leq a\leq m-1\\
\lceil r/a\rceil\leq n-1
}}
\left(
a+\left\lceil\frac{r}{a}\right\rceil
\right).
\]
Values $a>r$ cannot minimize this expression, so the two minimum
coincide. Consequently,
\[
d_r(\mathcal{C}_{m,n})
=
r+1+\lambda_{m,n}(r).
\]
\end{proof}

For  
$\mathcal{C}_n
=\operatorname{SPC}(n)\otimes\operatorname{SPC}(n)$,
Theorem~\ref{thm:ghw-spc-product} gives the following result.

\begin{corollary}\label{cor:ghw-square-spc-product}
For $1\leq r\leq(n-1)^2$,  the $r$-th generalized Hamming weight of
$\mathcal{C}_{n}$ is
\[
d_r(\mathcal{C}_n)
=
r+1+
\min_{\substack{
1\leq a,b\leq n-1\\
ab\geq r
}}
(a+b).
\]
Equivalently,
\[
d_r(\mathcal{C}_n)
=
r+1+
\min_{\left\lceil r/(n-1)\right\rceil\leq a\leq n-1}
\left(
a+\left\lceil\frac{r}{a}\right\rceil
\right).
\]
\end{corollary}

\begin{example}\label{ex:ghw-spc44}
For
$\mathcal{C}_{4}
=\operatorname{SPC}(4)\otimes\operatorname{SPC}(4)$,
the values of $\lambda_{4}(r)$, for $1\leq r\leq9$, are
\[
(2,3,4,4,5,5,6,6,6).
\]
Therefore, its generalized Hamming weight hierarchy is
\[
\bigl(
d_1(\mathcal{C}_{4}),\ldots,d_9(\mathcal{C}_{4})
\bigr)
=
(4,6,8,9,11,12,14,15,16).
\]

In particular, $d_4(\mathcal{C}_{4})=9$. This value is attained by
the subcode consisting of the matrices whose last row and last column
are zero and whose upper-left $3\times3$ submatrix belongs to
$\mathcal{C}_{3}$. This subcode has dimension $4$ and support size
$9$.
\end{example}

\begin{example}\label{ex:ghw-spc33}
For
$\mathcal{C}_{3}
=\operatorname{SPC}(3)\otimes\operatorname{SPC}(3)$,
we have
$
\bigl(
\lambda_{3,3}(1),
\lambda_{3,3}(2),
\lambda_{3,3}(3),
\lambda_{3,3}(4)
\bigr)
=
(2,3,4,4).
$
Thus,
\[
\bigl(
d_1(\mathcal{C}_{3}),
d_2(\mathcal{C}_{3}),
d_3(\mathcal{C}_{3}),
d_4(\mathcal{C}_{3})
\bigr)
=
(4,6,8,9).
\]
\end{example}

\color{black}
\subsection{Maximum weight and symmetry of the weight enumerator}

We conclude this section with two properties of the Hamming weight
distribution of the   product code $\mathcal{C}_n$. These
properties will also provide useful consistency checks for the explicit
weight enumerator derived later.

We begin by determining the maximum possible weight of a codeword. Recall that the minimum weight of a nonzero codeword is 4.
\begin{proposition}\label{prop:maximum-weight-spc-product}
Let
$
\mathcal{C}_n
=
\operatorname{SPC}(n)\otimes\operatorname{SPC}(n).
$
The maximum Hamming weight of a codeword in $\mathcal{C}_n$ is
\[
\max_{\mathbf{c}\in\mathcal{C}_n}
\operatorname{wt}(\mathbf{c})
=
\begin{cases}
n^2,   & \text{if $n$ is even},\\
n^2-n, & \text{if $n$ is odd}.
\end{cases}
\]
\end{proposition}

\begin{proof}
Every codeword of $\mathcal{C}_n$ can be represented as an $n\times n$
binary matrix whose rows and columns have even Hamming weight.

Suppose first that $n$ is even, and let $J_n$ be the $n\times n$
all-ones matrix. Every row and every column of $J_n$ has weight $n$,
which is even. Therefore,
$
J_n\in\mathcal{C}_n.
$
Since
$
\operatorname{wt}(J_n)=n^2
$
and $n^2$ is the length of the code, it follows that
\[
\max_{\mathbf{c}\in\mathcal{C}_n}
\operatorname{wt}(\mathbf{c})
=n^2.
\]

Now suppose that $n$ is odd. Each row of a codeword in
$\mathcal{C}_n$ has even weight and therefore cannot have weight $n$.
Thus, every row has weight at most $n-1$, which gives
\[
\operatorname{wt}(\mathbf{c})
\leq
n(n-1)
=
n^2-n
\]
for every $\mathbf{c}\in\mathcal{C}_n$.

To show that this bound is attained, consider
\[
X=J_n+I_n,
\]
where the addition is performed over $\mathbb{F}_2$. The matrix $X$
has zeros on the main diagonal and ones elsewhere. Consequently, every
row and every column has weight $n-1$, which is even because $n$ is
odd. Hence,
$
X\in\mathcal{C}_n,
$
and
\[
\operatorname{wt}(X)
=
n(n-1)
=
n^2-n.
\]
Therefore, the upper bound is attained.
\end{proof}

The next result establishes that the weight enumerator is symmetric when $n$ is even.

\begin{proposition}\label{prop:symmetry-even-n}
Let $\mathcal{C}_n$ have weight distribution
$(A_0,A_1,\ldots,A_{n^2})$. If $n$ is even, then
\[
A_t=A_{n^2-t}
\]
for every $0\leq t\leq n^2$. Equivalently, its homogeneous weight
enumerator satisfies
\[
W_{\mathcal{C}_n}(x,y)
=
W_{\mathcal{C}_n}(y,x).
\]
\end{proposition}

\begin{proof}
Since $n$ is even, the all-one vector $\mathbf{1}_n$ belongs to
$\operatorname{SPC}(n)$. Therefore,
\[
\mathbf{1}_{n^2}
=
\mathbf{1}_n\otimes\mathbf{1}_n
\in
\mathcal{C}_n.
\]

Let $\mathbf{c}\in\mathcal{C}_n$ have weight $t$. Since
$\mathcal{C}_n$ is linear,
$
\mathbf{c}+\mathbf{1}_{n^2}
\in
\mathcal{C}_n.
$
Addition by the all-one vector complements every coordinate, and hence
\[
\operatorname{wt}\bigl(\mathbf{c}+\mathbf{1}_{n^2}\bigr)
=
n^2-t.
\]
Moreover, the map
\[
\mathbf{c}
\longmapsto
\mathbf{c}+\mathbf{1}_{n^2}
\]
is an involution and therefore defines a bijection between the
codewords of weight $t$ and those of weight $n^2-t$. Consequently,
\[
A_t=A_{n^2-t}
\]
for every $0\leq t\leq n^2$. The result follows.
\end{proof}

Combining Propositions~\ref{prop:maximum-weight-spc-product} and
\ref{prop:symmetry-even-n} gives the following characterization.

\begin{corollary}\label{cor:symmetric-weight-enumerator-spc-product}
The homogeneous weight enumerator of $\mathcal{C}_n$ is symmetric if
and only if $n$ is even.
\end{corollary}

\begin{proof}
If $n$ is even, the result follows from
Proposition~\ref{prop:symmetry-even-n}. If $n$ is odd, then
Proposition~\ref{prop:maximum-weight-spc-product} gives
$
A_{n^2}=0,
$
whereas $A_0=1$. Hence, the weight enumerator cannot be symmetric.
\end{proof}

\section{The Weight Enumerator of the SPC Product Code via the
Walsh--Hadamard Transform}


In this section, we derive an explicit expression for the weight enumerator
of the product of two binary SPC codes. The derivation is
based on the description of the dual code and the MacWilliams identity,
or equivalently, on the Walsh--Hadamard transform.

For integers $m,n\geq 2$, let $\mathcal{C}_{m,n}$ denote the binary code
consisting of all $m\times n$ matrices whose rows and columns have even
Hamming weight. Up to a permutation of coordinates, this code is
$
\mathcal{C}_{m,n}
=
\operatorname{SPC}(m)\otimes\operatorname{SPC}(n)
$ and has parameters
$
[mn,(m-1)(n-1),4].
$

For $\mathbf{a}\in\mathbb{F}_2^m$ and
$\mathbf{b}\in\mathbb{F}_2^n$, define
\[
Y(\mathbf{a},\mathbf{b})
=
\bigl(a_i+b_j\bigr)_
{\substack{1\leq i\leq m\\1\leq j\leq n}}
\in\mathbb{F}_2^{m\times n}.
\]

\begin{lemma}\label{lem:dual-spc-product}
The dual code of $\mathcal{C}_{m,n}$ is
\[
\mathcal{C}_{m,n}^{\perp}
=
\left\{
Y(\mathbf{a},\mathbf{b}) :
\mathbf{a}\in\mathbb{F}_2^m,\ 
\mathbf{b}\in\mathbb{F}_2^n
\right\}.
\]
Moreover, every codeword of $\mathcal{C}_{m,n}^{\perp}$ has exactly two
representations of this form.
\end{lemma}


\begin{proof}
Consider the linear map
\[
\begin{aligned}
\Phi:
\mathbb{F}_2^m\times\mathbb{F}_2^n
&\longrightarrow
\mathbb{F}_2^{m\times n},\\
(\mathbf{a},\mathbf{b})
&\longmapsto
Y(\mathbf{a},\mathbf{b}).
\end{aligned}
\]
We first prove that
$
\operatorname{Im}(\Phi)
\subseteq
\mathcal{C}_{m,n}^{\perp}.
$
Let $X=[x_{ij}]\in\mathcal{C}_{m,n}$. Since every row and every
column of $X$ has even weight, we have
\[
\sum_{j=1}^{n}x_{ij}=0
\quad\text{for every }1\leq i\leq m,
\]
\[
\sum_{i=1}^{m}x_{ij}=0
\quad\text{for every }1\leq j\leq n.
\]
All these equalities are taken in $\mathbb{F}_2$.

For any
$\mathbf{a}=(a_1,\ldots,a_m)\in\mathbb{F}_2^m$ and
$\mathbf{b}=(b_1,\ldots,b_n)\in\mathbb{F}_2^n$, we obtain
\[
\begin{aligned}
\left\langle
Y(\mathbf{a},\mathbf{b}),X
\right\rangle
&=
\sum_{i=1}^{m}\sum_{j=1}^{n}
(a_i+b_j)x_{ij}=
\sum_{i=1}^{m}
a_i\left(\sum_{j=1}^{n}x_{ij}\right)
+
\sum_{j=1}^{n}
b_j\left(\sum_{i=1}^{m}x_{ij}\right)\\
&=0.
\end{aligned}
\]
Therefore, $Y(\mathbf{a},\mathbf{b})$ is orthogonal to every codeword
of $\mathcal{C}_{m,n}$, and hence
$
\operatorname{Im}(\Phi)
\subseteq
\mathcal{C}_{m,n}^{\perp}.
$

We now determine the dimension of $\operatorname{Im}(\Phi)$. If
$(\mathbf{a},\mathbf{b})\in\ker(\Phi)$, then
$
a_i+b_j=0
$
for every $i$ and $j$. Thus, all the coordinates of $\mathbf{a}$ and
$\mathbf{b}$ are equal to the same element of $\mathbb{F}_2$.
Consequently,
\[
\ker(\Phi)
=
\left\{
(\mathbf{0},\mathbf{0}),
(\mathbf{1}_m,\mathbf{1}_n)
\right\}.
\]
In particular, $\dim(\ker(\Phi))=1$, and the rank--nullity theorem
gives
\[
\dim\bigl(\operatorname{Im}(\Phi)\bigr)
=
m+n-1.
\]
Since $\mathcal{C}_{m,n}$ has dimension $(m-1)(n-1)$, its dual has
dimension
\[
\begin{aligned}
\dim\bigl(\mathcal{C}_{m,n}^{\perp}\bigr)
&=
mn-(m-1)(n-1)=
m+n-1.
\end{aligned}
\]
Therefore,
$\operatorname{Im}(\Phi)
=
\mathcal{C}_{m,n}^{\perp}.
$

Finally, every fiber of $\Phi$ is a coset of $\ker(\Phi)$ and hence
contains exactly two elements. More explicitly, the two
representations of a dual codeword are
\[
(\mathbf{a},\mathbf{b})
\quad\text{and}\quad
(\mathbf{a}+\mathbf{1}_m,\mathbf{b}+\mathbf{1}_n).
\]
Thus, every codeword of $\mathcal{C}_{m,n}^{\perp}$ has exactly two
representations of the stated form.
\end{proof}

For $0\leq r\leq m$ and $0\leq t\leq n$, define
\begin{align*}
Z(r,t) &= rt+(m-r)(n-t), \\
O(r,t) &= r(n-t)+(m-r)t.
\end{align*}
Notice that
$
Z(r,t)+O(r,t)=mn.
$

If $\operatorname{wt}(\mathbf{a})=r$ and $\operatorname{wt}(\mathbf{b})=t$, then
$Z(r,t)$ is the number of pairs $(i,j)$ for which $a_i=b_j$, whereas
$O(r,t)$ is the number of pairs for which $a_i\neq b_j$. Consequently,
\[
\operatorname{wt}\bigl(Y(\mathbf{a},\mathbf{b})\bigr)=O(r,t).
\]

\begin{theorem}\label{thm:spc-product-weight-enumerator}
The homogeneous weight enumerator of $\mathcal{C}_{m,n}$ is
\[
W_{\mathcal{C}_{m,n}}(x,y)
=
\frac{1}{2^{m+n}}
\sum_{r=0}^{m}\binom{m}{r}
\sum_{t=0}^{n}\binom{n}{t}
(x+y)^{Z(r,t)}
(x-y)^{O(r,t)}.
\]
Equivalently, its one-variable weight enumerator is
\[
W_{\mathcal{C}_{m,n}}(z)
=
\frac{1}{2^{m+n}}
\sum_{r=0}^{m}\binom{m}{r}
\sum_{t=0}^{n}\binom{n}{t}
(1+z)^{Z(r,t)}
(1-z)^{O(r,t)}.
\]
\end{theorem}

\begin{proof}
By Lemma~\ref{lem:dual-spc-product}, every word of
$\mathcal{C}_{m,n}^{\perp}$ is represented exactly twice by a pair
$(\mathbf{a},\mathbf{b})$. Therefore,
\[
W_{\mathcal{C}_{m,n}^{\perp}}(x,y)
=
\frac{1}{2}
\sum_{\mathbf{a}\in\mathbb{F}_2^m}
\sum_{\mathbf{b}\in\mathbb{F}_2^n}
x^{mn-\operatorname{wt}(Y(\mathbf{a},\mathbf{b}))}
y^{\operatorname{wt}(Y(\mathbf{a},\mathbf{b}))}.
\]
Grouping the vectors according to
$
r=\operatorname{wt}(\mathbf{a}), 
t=\operatorname{wt}(\mathbf{b}),
$
gives
\[
W_{\mathcal{C}_{m,n}^{\perp}}(x,y)
=
\frac{1}{2}
\sum_{r=0}^{m}\binom{m}{r}
\sum_{t=0}^{n}\binom{n}{t}
x^{Z(r,t)}y^{O(r,t)}.
\]
Since
$
\left|\mathcal{C}_{m,n}^{\perp}\right|
=
2^{m+n-1},
$
the MacWilliams identity gives
\[
W_{\mathcal{C}_{m,n}}(x,y)
=
\frac{1}{\left|\mathcal{C}_{m,n}^{\perp}\right|}
W_{\mathcal{C}_{m,n}^{\perp}}(x+y,x-y).
\]
Substitution yields
\[
W_{\mathcal{C}_{m,n}}(x,y)
=
\frac{1}{2^{m+n}}
\sum_{r=0}^{m}\binom{m}{r}
\sum_{t=0}^{n}\binom{n}{t}
(x+y)^{Z(r,t)}
(x-y)^{O(r,t)}.
\]
Finally, setting $x=1$ and $y=z$ gives the one-variable expression.
\end{proof}

\begin{corollary}\label{cor:spc-product-coefficients}
Write
\[
W_{\mathcal{C}_{m,n}}(z)
=
\sum_{w=0}^{mn}A_wz^w.
\]
Then, for every $0\leq w\leq mn$,
\[
A_w
=
\frac{1}{2^{m+n}}
\sum_{r=0}^{m}\binom{m}{r}
\sum_{t=0}^{n}\binom{n}{t}
\sum_{u=\max\{0,w-Z(r,t)\}}^{\min\{w,O(r,t)\}}
(-1)^u
\binom{Z(r,t)}{w-u}
\binom{O(r,t)}{u}.
\]
\end{corollary}

\begin{proof}
For fixed $r$ and $t$, the coefficient of $z^w$ in
$
(1+z)^{Z(r,t)}(1-z)^{O(r,t)}
$
is
\[
\sum_{u=\max\{0,w-Z(r,t)\}}^{\min\{w,O(r,t)\}}
(-1)^u
\binom{Z(r,t)}{w-u}
\binom{O(r,t)}{u}.
\]
Substituting this coefficient into
Theorem~\ref{thm:spc-product-weight-enumerator} proves the result.
\end{proof}

For $
\mathcal{C}_n
=
\operatorname{SPC}(n)\otimes\operatorname{SPC}(n),$
we have
\begin{align*}
Z_n(r,t) &= rt+(n-r)(n-t), \\
O_n(r,t) &= r(n-t)+(n-r)t.
\end{align*}
Hence,
\[
W_{\mathcal{C}_n}(z)
=
\frac{1}{2^{2n}}
\sum_{r=0}^{n}\binom{n}{r}
\sum_{t=0}^{n}\binom{n}{t}
(1+z)^{Z_n(r,t)}
(1-z)^{O_n(r,t)}.
\]
\begin{corollary}\label{cor:square-spc-weight-enumerator}
For
$\mathcal{C}_n
=\operatorname{SPC}(n)\otimes\operatorname{SPC}(n)$,
we have
\[
W_{\mathcal{C}_n}(z)
=
\frac{1}{2^{2n}}
\sum_{r=0}^{n}\binom{n}{r}
\sum_{t=0}^{n}\binom{n}{t}
(1+z)^{Z_n(r,t)}
(1-z)^{O_n(r,t)}.
\]
\end{corollary}
\begin{example}
Consider
$\mathcal{C}_3
=\operatorname{SPC}(3)\otimes\operatorname{SPC}(3)$,
which has parameters $[9,4,4]_2$. We compute the coefficient $A_4$.
For convenience, let
\[
S_4(Z,O)
=
[z^4](1+z)^Z(1-z)^O
\]
and define
\[
\mu_Z
=
\sum_{\substack{0\leq r,t\leq 3\\Z(r,t)=Z}}
\binom{3}{r}\binom{3}{t}.
\]
The relevant values are
\[
\begin{array}{c|c|c|c|c}
Z & O & \mu_Z & S_4(Z,O) & \mu_ZS_4(Z,O)\\
\hline
9 & 0 & 2  & 126 & 252\\
0 & 9 & 2  & 126 & 252\\
6 & 3 & 12 & -6  & -72\\
3 & 6 & 12 & -6  & -72\\
5 & 4 & 18 & 6   & 108\\
4 & 5 & 18 & 6   & 108
\end{array}
\]
Therefore,
\[
A_4
=
\frac{252+252-72-72+108+108}{2^6}
=
\frac{576}{64}
=
9.
\]
Applying the formula to all possible weights gives
\[
W_{\mathcal{C}_3}(z)
=
1+9z^4+6z^6.
\]
\end{example}

The expression in Theorem~\ref{thm:spc-product-weight-enumerator}
replaces an enumeration over
$2^{(m-1)(n-1)}$ codewords by a sum indexed only by the possible weights
of two auxiliary binary vectors.


\section{Computational Evaluation of the Weight Enumerator}

The closed-form expression obtained in
Theorem~\ref{thm:spc-product-weight-enumerator} determines every
coefficient of the weight enumerator of
$
\mathcal{C}_{m,n}
=
\operatorname{SPC}(m)\otimes\operatorname{SPC}(n).$
More precisely,
\[
A_w
=
\frac{1}{2^{m+n}}
\sum_{r=0}^{m}\binom{m}{r}
\sum_{t=0}^{n}\binom{n}{t}
\sum_{u=\max\{0,w-Z(r,t)\}}^{\min\{w,O(r,t)\}}
(-1)^u
\binom{Z(r,t)}{w-u}
\binom{O(r,t)}{u},
\]
where
\begin{align*}
Z(r,t) &= rt+(m-r)(n-t), \\
O(r,t) &= r(n-t)+(m-r)t.
\end{align*}

Although this expression avoids the enumeration of the
$2^{(m-1)(n-1)}$ codewords, a direct implementation still performs
many repeated calculations. Indeed, different pairs $(r,t)$ may produce
the same value of $O(r,t)$. These terms can be grouped before computing
the coefficients.

Let $N=mn$ and, for $0\leq j\leq N$, define
\[
\mu_j
=
\sum_{\substack{0\leq r\leq m,\ 0\leq t\leq n\\O(r,t)=j}}
\binom{m}{r}\binom{n}{t}.
\]
The coefficient formula can then be written as
\[
A_w
=
\frac{1}{2^{m+n}}
\sum_{j=0}^{N}
\mu_j K_w(j;N),
\]
where
\[
K_w(j;N)
=
\sum_{u=\max\{0,w-N+j\}}^{\min\{w,j\}}
(-1)^u
\binom{N-j}{w-u}
\binom{j}{u}
\]
is the binary Krawtchouk polynomial of degree $w$, evaluated at $j$ \cite{Krawtchouk1929}.
Equivalently,
\[
K_w(j;N)
=
[z^w](1+z)^{N-j}(1-z)^j.
\]

The Krawtchouk polynomials satisfy the recurrence
\[
K_0(j;N)=1,
\qquad
K_1(j;N)=N-2j,
\]
and
\[
(w+1)K_{w+1}(j;N)
=
(N-2j)K_w(j;N)
-
(N-w+1)K_{w-1}(j;N)
\]
for $1\leq w\leq N-1$. This recurrence avoids evaluating the
alternating binomial sum separately for every coefficient.
Algorithm~\ref{alg:spc-product-enumerator} describes the resulting
procedure.

\begin{algorithm}
\caption{Exact weight enumerator of
$\mathcal{C}_{m,n}=\operatorname{SPC}(m)\otimes\operatorname{SPC}(n)$}
\label{alg:spc-product-enumerator}
\begin{algorithmic}[1]

\State \textbf{Input:} Integers $m,n\geq2$
\State \textbf{Output:} The coefficients
$\{A_w\}_{w=0}^{N}$ and $W_{\mathcal{C}_{m,n}}(x,y)$

\State $N\gets mn$
\State Initialize $\mu[j]\gets 0$ for $j=0,\ldots,N$
\State Initialize $T[w]\gets 0$ for $w=0,\ldots,N$

\Statex
\Comment{Group the pairs $(r,t)$ according to $O(r,t)$}
\For{$r=0$ \textbf{to} $m$}
    \For{$t=0$ \textbf{to} $n$}
        \State $Z\gets rt+(m-r)(n-t)$
        \State $O\gets N-Z$
        \State $\mu[O]\gets
        \mu[O]+\binom{m}{r}\binom{n}{t}$
    \EndFor
\EndFor

\Statex
\Comment{Evaluate the Krawtchouk polynomials recursively}
\For{$j=0$ \textbf{to} $N$}
    \If{$\mu[j]\neq 0$}
        \State $K_{\mathrm{prev}}\gets 1$
        \State $T[0]\gets T[0]+\mu[j]K_{\mathrm{prev}}$

        \If{$N\geq 1$}
            \State $K_{\mathrm{curr}}\gets N-2j$
            \State $T[1]\gets T[1]+\mu[j]K_{\mathrm{curr}}$

            \For{$w=1$ \textbf{to} $N-1$}
                \State
                $K_{\mathrm{next}}\gets
                \dfrac{
                (N-2j)K_{\mathrm{curr}}
                -(N-w+1)K_{\mathrm{prev}}
                }{w+1}$

                \State
                $T[w+1]\gets
                T[w+1]+\mu[j]K_{\mathrm{next}}$

                \State $K_{\mathrm{prev}}\gets K_{\mathrm{curr}}$
                \State $K_{\mathrm{curr}}\gets K_{\mathrm{next}}$
            \EndFor
        \EndIf
    \EndIf
\EndFor

\Statex
\Comment{Normalize the coefficients}
\For{$w=0$ \textbf{to} $N$}
    \State $A[w]\gets T[w]/2^{m+n}$
\EndFor

\Statex
\Comment{Construct the homogeneous weight enumerator}
\State $W_{\mathcal{C}_{m,n}}(x,y)\gets 0$

\For{$w=0$ \textbf{to} $N$}
    \If{$A[w]\neq 0$}
        \State
        $W_{\mathcal{C}_{m,n}}(x,y)
        \gets
        W_{\mathcal{C}_{m,n}}(x,y)
        +A[w]x^{N-w}y^w$
    \EndIf
\EndFor

\State \Return
$W_{\mathcal{C}_{m,n}}(x,y),\{A_w\}_{w=0}^{N}$

\end{algorithmic}
\end{algorithm}
After precomputing the binomial coefficients
$\binom{m}{r}$ and $\binom{n}{t}$, the multiplicities $\mu_j$ can be
computed using $O(mn)=O(N)$ arithmetic operations. The Krawtchouk
recurrence requires at most $O(N^2)$ additional arithmetic operations,
and the algorithm uses $O(N)$ auxiliary memory. These estimates count
exact integer arithmetic operations; since the coefficients can be
very large, arbitrary-precision integer arithmetic must be used.

Several properties provide consistency checks for the implementation.
First, the multiplicities satisfy
\[
\sum_{j=0}^{N}\mu_j
=
2^{m+n},
\]
since they count all pairs
$(\mathbf{a},\mathbf{b})\in
\mathbb{F}_2^m\times\mathbb{F}_2^n$.
Moreover, every $\mu_j$ is even because the pairs
$
(\mathbf{a},\mathbf{b})
$ and $
(\mathbf{a}+\mathbf{1}_m,\mathbf{b}+\mathbf{1}_n)
$
produce the same dual codeword.

For the resulting weight distribution, we must have
$
A_0=1
$
and
\[
\sum_{w=0}^{N}A_w
=
2^{(m-1)(n-1)}.
\]
Moreover, every codeword has even weight, and hence
$
A_w=0
$
whenever $w$ is odd. Finally, if both $m$ and $n$ are even, the
all-one matrix belongs to $\mathcal{C}_{m,n}$, and therefore
\[
A_w=A_{N-w}.
\]


\subsection{Computational examples}
\subsubsection*{The code
$\operatorname{SPC}(3)\otimes\operatorname{SPC}(3)$}

This code has parameters $[9,4,4]_2$ and weight enumerator
\[
W_{\mathcal{C}_{3}}(z)
=
1+9z^4+6z^6.
\]
In particular,
$
A_0=1,
A_4=9,
A_6=6,$
and
$ 1+9+6=16=2^4.$

\subsubsection*{The code
$\operatorname{SPC}(4)\otimes\operatorname{SPC}(4)$}

This code has parameters $[16,9,4]_2$ and weight enumerator
\[
\begin{aligned}
W_{\mathcal{C}_{4}}(z)
={}&
1+36z^4+96z^6+246z^8+96z^{10}+36z^{12}+z^{16}.
\end{aligned}
\]
The sum of its coefficients is
\[
1+36+96+246+96+36+1
=
512
=
2^9.
\]
Since $4$ is even,  the
coefficients satisfy
$
A_w=A_{16-w}.$

\subsubsection*{The code
$\operatorname{SPC}(10)\otimes\operatorname{SPC}(10)$}

This code has parameters $[100,81,4]_2$. Some of the coefficients
obtained by Algorithm~\ref{alg:spc-product-enumerator} are shown in
Table~\ref{tab:spc10-coefficients}.

\begin{table}[htbp]
\centering
\caption{Selected coefficients of the weight enumerator of
$\operatorname{SPC}(10)\otimes\operatorname{SPC}(10)$.}
\label{tab:spc10-coefficients}
\begin{tabular}{c|r|c|r}
\toprule
$w$ & $A_w$ & $100-w$ & $A_{100-w}$\\
\midrule
0  & $1$                 & 100 & $1$\\
4  & $2{,}025$           & 96  & $2{,}025$\\
6  & $86{,}400$          & 94  & $86{,}400$\\
8  & $4{,}895{,}100$     & 92  & $4{,}895{,}100$\\
10 & $213{,}615{,}360$   & 90  & $213{,}615{,}360$\\
12 & $7{,}987{,}574{,}700$
   & 88 & $7{,}987{,}574{,}700$\\
14 & $246{,}565{,}468{,}800$
   & 86 & $246{,}565{,}468{,}800$\\
16 & $6{,}238{,}507{,}153{,}050$
   & 84 & $6{,}238{,}507{,}153{,}050$\\
\bottomrule
\end{tabular}
\end{table}

The symmetry $A_w=A_{100-w}$ follows from the fact that $n$ is even (see Corollary~\ref{cor:symmetric-weight-enumerator-spc-product}).  Moreover,
the complete distribution satisfies
\[
\sum_{w=0}^{100}A_w
=
2^{81},
\]
providing an additional verification of the computation.

\section{Conclusion}

In this paper, we investigated structural and enumerative properties of
the binary single parity-check product codes
$
\mathcal{C}_{m,n}
=
\operatorname{SPC}(m)\otimes\operatorname{SPC}(n).$

By applying the product formula for codes satisfying the chain
condition, we obtained an explicit description of their generalized
Hamming weight hierarchy. For the square product
$\mathcal{C}_n$, we also determined the maximum
codeword weight and proved that its homogeneous weight enumerator is
symmetric if and only if $n$ is even.

The main result is an exact closed-form expression for the weight
enumerator of $\mathcal{C}_{m,n}$. After characterizing the dual code,
we applied the MacWilliams identity in its Walsh--Hadamard formulation.
The dual codewords were parametrized by two auxiliary binary vectors,
and grouping the corresponding terms according to the Hamming weights
of these vectors led to explicit formulas for both the weight
enumerator and each of its coefficients.

The structure of the coefficient formula also yields an exact
computational procedure based on Krawtchouk polynomials, avoiding the
exhaustive enumeration of the
$2^{(m-1)(n-1)}$ codewords. The numerical examples and consistency
checks verify the resulting weight distributions for several component
lengths.

The Walsh--Hadamard approach developed here provides a unified
framework for studying the finite-length weight distribution of
two-dimensional SPC product codes. Possible directions for further
research include extending the method to higher-dimensional SPC
products, investigating analogous formulas for other component codes,
and further exploring the connection between the coefficient formula
and Krawtchouk transforms.

\section*{Acknowledgments}
The first author was supported by Conselho Nacional de Desenvolvimento Científico e Tecnológico (CNPq), Brazil with number of process 141533/2021-8.
The second author  was also supported by CNPq with process 405842/2023-6 and by FAPESP with processes 2024/05051-7 and 2024/00923-6.

\bibliography{sn-bibliography} 
\end{document}